\documentclass[a4paper,UKenglish,cleveref, autoref, thm-restate]{lipics-v2021}

\pdfoutput=1 
\hideLIPIcs  

\title{Finding the saddlepoint faster than sorting}

\titlerunning{Finding the saddlepoint faster than sorting} 

\author{Justin Dallant}{Department of Computer Science, Université libre de Bruxelles, Belgium} {Justin.Dallant@ulb.be}{https://orcid.org/0000-0001-5539-9037}{Supported by the French Community of Belgium via the funding of a FRIA grant.}

\author{Frederik Haagensen
}{Department of Computer Science, IT University of Copenhagen, Denmark}{haag@itu.dk}{https://orcid.org/0000-0002-4161-4442}{Supported by Independent Research Fund Denmark, grant 0136-00144B, ``DISTRUST'' project.}

\author{Riko Jacob}{Department of Computer Science, IT University of Copenhagen, Denmark}{rikj@itu.dk}{https://orcid.org/0000-0001-9470-1809}{}

\author{László Kozma}{Institut für Informatik, Freie Universität Berlin, Germany}{laszlo.kozma@fu-berlin.de}{https://orcid.org/0000-0002-3253-2373}{Supported by DFG grant
KO 6140/1-2.}

\author{Sebastian Wild}{Department of Computer Science, University of Liverpool, UK}{wild@liverpool.ac.uk}{https://orcid.org/0000-0002-6061-9177}{}

\authorrunning{J. Dallant,  F. Haagensen, R. Jacob, L. Kozma and S. Wild} 

\Copyright{Justin Dallant, Frederik Haagensen, Riko Jacob, László Kozma, and Sebastian Wild} 



\category{} 

\relatedversion{} 

\acknowledgements{This work was initiated at Dagstuhl Seminar 23211 ``Scalable Data Structures''.}

\nolinenumbers 

\EventEditors{}
\EventNoEds{2}
\EventLongTitle{}
\EventShortTitle{arXiv 2023}
\EventAcronym{arXiv}
\EventYear{2023}
\EventDate{}
\EventLocation{}
\EventLogo{}
\SeriesVolume{}
\ArticleNo{}

\usepackage{mathtools}
\newcommand{\DeclareAutoPairedDelimiter}[3]{%
  \expandafter\DeclarePairedDelimiter\csname Auto\string#1\endcsname{#2}{#3}%
  \begingroup\edef\x{\endgroup
    \noexpand\DeclareRobustCommand{\noexpand#1}{%
      \expandafter\noexpand\csname Auto\string#1\endcsname*}}%
  \x}

\DeclareAutoPairedDelimiter\ceil{\lceil}{\rceil}
\DeclareAutoPairedDelimiter\floor{\lfloor}{\rfloor}
\DeclarePairedDelimiter\pars{(}{)}
\newcommand{\OO}[1]{O\pars*{#1}}
\renewcommand{\epsilon}{\varepsilon}

\usepackage{mdframed}  

\usepackage{thmtools}
\usepackage{thm-restate}

\usepackage{tikz}
\usepackage{algpseudocode}

\begin{document}

\maketitle

\begin{abstract}
A \emph{saddlepoint} of an $n \times n$ matrix $A$ is an entry of $A$ that is a maximum in its row and a minimum in its column. Knuth (1968) gave several different algorithms for finding a saddlepoint. The worst-case running time of these algorithms is $\Theta(n^2)$, and Llewellyn, Tovey, and Trick (1988) showed that this cannot be improved, as in the worst case all entries of $A$ may need to be queried.

A \emph{strict saddlepoint} of $A$ is an entry that is the strict maximum in its row and the strict minimum in its column. The strict saddlepoint (if it exists) is unique, and Bienstock, Chung, Fredman, Schäffer, Shor, and Suri (1991) showed that it can be found in time $\OO{n \lg{n}}$, 
where a dominant runtime contribution is sorting the diagonal of the matrix. This upper bound has not been improved since 1991. 
In this paper we show that the strict saddlepoint can be found in $\OO{n \lg^*{n}} \subset o(n\lg n)$ time, where $\lg^*$ denotes the very slowly growing \emph{iterated logarithm} function, coming close to the lower bound of $\Omega{(n)}$. In fact, we can also compute, within the same runtime, the \emph{value} of a non-strict saddlepoint, assuming one exists. Our algorithm is based on a simple recursive approach, a feasibility test inspired by searching in sorted matrices, and a relaxed notion of saddlepoint. 

\end{abstract}

\section{Introduction}

Saddlepoints are a central concept of mathematical analysis and numerical optimization. Informally, a saddlepoint of a function (or a surface) is a point where the derivatives (slopes) in orthogonal directions vanish, yet the point is not a local minimum or maximum. In this paper we are concerned with a discrete analogue: an entry of a matrix $A$ is a saddlepoint, if it is simultaneously the maximum in its row and the minimum in its column. 

\begin{figure}[htb]
 
 \begin{minipage}[c]{.4\linewidth}
		\includegraphics[width=\linewidth]{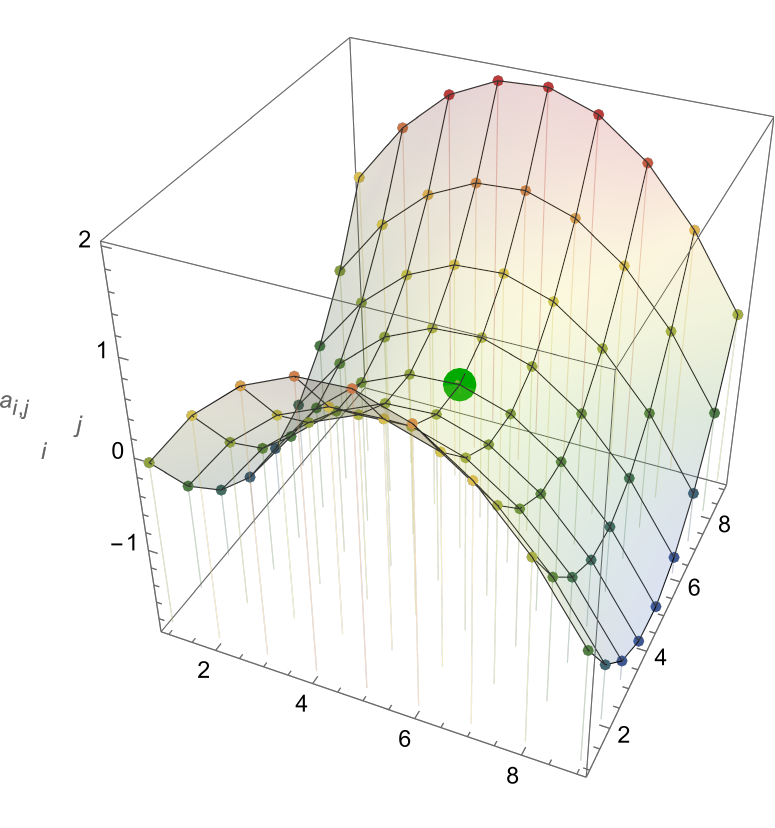}
	\end{minipage}%
	\begin{minipage}[c]{.4\linewidth}
	~\\[0\baselineskip]
	\hspace*{-1em}\resizebox{1.45\linewidth}!{\quad\quad~~~~~~~~~~~~$A=\left(
	\begin{array}{ccccccccc}
	 -0.08 & 0.55 & 0.98 & 1.21 & 1.24 & 1.07 & 0.7 & 0.13 & -0.64 \\
	 -0.69 & -0.06 & 0.37 & 0.6 & 0.63 & 0.46 & 0.09 & -0.48 & -1.25 \\
	 -1.1 & -0.47 & -0.04 & 0.19 & 0.22 & 0.05 & -0.32 & -0.89 & -1.66 \\
	 -1.31 & -0.68 & -0.25 & -0.02 & 0.01 & -0.16 & -0.53 & -1.1 & -1.87 \\
	 -1.32 & -0.69 & -0.26 & -0.03 & 0 & -0.17 & -0.54 & -1.11 & -1.88 \\
	 -1.13 & -0.5 & -0.07 & 0.16 & 0.19 & 0.02 & -0.35 & -0.92 & -1.69 \\
	 -0.74 & -0.11 & 0.32 & 0.55 & 0.58 & 0.41 & 0.04 & -0.53 & -1.3 \\
	 -0.15 & 0.48 & 0.91 & 1.14 & 1.17 & 1.0 & 0.63 & 0.06 & -0.71 \\
	 0.64 & 1.27 & 1.7 & 1.93 & 1.96 & 1.79 & 1.42 & 0.85 & 0.08 \\
	\end{array}
	\right)$}
	\end{minipage}
	\caption{(\emph{right}) A $9\times 9$-matrix $A$ with a (strict) saddlepoint at $a_{5,5}=0$. (\emph{left}) A 3D plot of the matrix entries with the saddlepoint highlighted in green.}
	\label{}
\end{figure}

If $A$ represents the payoff matrix of a two-player zero-sum game, then saddlepoints of~$A$ give the \emph{value} of the game, corresponding exactly to the pure-strategy Nash equilibria (see e.g.~\cite[\S\,4]{gt}). Thus, finding the saddlepoint of a matrix (as fast as possible) is a natural and fundamental algorithmic question. 

Knuth considered the saddlepoint problem~\cite[\S\,1.3.2]{Knuth1} already in 1968, observing that saddlepoints of $A$ (if they exist) must equal both the minimum of all row-maxima and the maximum of all column-minima, and thus, all saddlepoints of a matrix have the same value. 
Knuth also gave a number of algorithms~\cite[pg.~512--515]{Knuth1} for finding saddlepoints (or reporting their absence). The runtimes of these algorithms may differ significantly on concrete instances, yet in the worst case they all perform $\Theta(n^2)$ operations on an $n \times n$ square input matrix $A$.\footnote{A running time of the form $\OO{n\cdot f(n)}$ for $n \times n$ input matrices also implies a running time of $\OO{m\cdot f(n)}$ for $m \times n$ or $n \times m$ input matrices with $m \geq n$, as shown in~\S\,\ref{sec2}.} In fact, such a runtime is necessary, as in the worst case all entries of $A$ must be inspected; this can be seen through a simple adversary argument, as shown by Llewellyn, Tovey, and Trick~\cite{Llewellyn1988}.

The situation changes considerably if we require the saddlepoint to be the \emph{strict maximum} in its row and the \emph{strict minimum} in its column: we refer to such an entry as a \emph{strict saddlepoint}. Note that we cannot transform the first problem into the second through a simple perturbation: adding noise to a matrix with non-strict saddlepoints may well create a matrix with no saddlepoints at all! It is not hard to see (e.g.\ by the above observation of Knuth), that at most one entry of a matrix can be a strict saddlepoint.

The first nontrivial algorithm for finding a strict saddlepoint was given by Llewellyn, Tovey, and Trick~\cite{Llewellyn1988} in 1988, with a runtime of $\OO{n^{\lg{3}}} \subset \OO{n^{1.59}}$, which they conjectured to be essentially optimal\footnote{Throughout the paper, $\lg{x}$ denotes the base-2 logarithm of $x$.}, see also~\cite{Hedet}. This conjecture turned out to be false, and an algorithm with runtime $\OO{n \lg{n}}$ was given by Bienstock, Chung, Fredman, Schäffer, Shor, and Suri~\cite{Bienstock1991} in 1991. Independently and around the same time, Byrne and Vaserstein~\cite{Byrne1991} obtained a similar result.

The bound of $\OO{n\lg{n}}$ on the complexity of the problem is a natural barrier and the algorithms achieving it are surprisingly simple.\footnote{We review the algorithm of Bienstock et al.~in a slightly modified form in \S\,\ref{sec3}.} They are deterministic, operate on the matrix~$A$ only via constant-time comparisons of entries, and need to inspect only $\OO{n}$ (adaptively queried) entries. In its original presentation, the algorithm of Bienstock et al.~involves $\Theta(n)$ operations on a \emph{heap} that holds $\Theta(n)$ keys -- alternatively it can be seen as executing an initial sorting step on $\Theta(n)$ entries that dominates the runtime. 

Algorithms based on such a step clearly cannot avoid making at least $\Theta(n \lg{n})$ comparisons. But is sorting necessary for finding a saddlepoint? Informally, sorting and saddlepoint-finding do seem related, although in a non-obvious way. If the input matrix $A$ does \emph{not} have a strict saddlepoint, then any correct algorithm must certify this (at least implicitly) for each entry $q$ of~$A$. A sufficient certificate for $q$ is a pair $(x,y)$  with $x \geq y$, where $x$ is an entry in the same row as $q$, and $y$ is an entry in the same column as $q$ (one of $x$ and $y$ may be $q$ itself); see Figure~\ref{fig:certificate}. A moment of thought reveals that such a certificate is also necessary: without it, $q$ may still be the strict saddlepoint. 

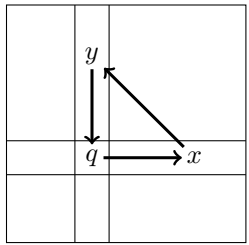
\begin{figure}[htbp]
\centering
	\begin{tikzpicture}[scale=.45]
		\def\n{7}
		\draw (0,0) rectangle (\n,\n) ;
		\foreach \x in {2,3} {
			\draw (\x,0) -- (\x,\n);
		}
		\foreach \y in {2,3} {
			\draw (0,\y) -- (\n,\y);
		}
		\node[inner sep=1.7pt] (q) at (2.5,2.5) {$q$} ;
		\node[inner sep=1.7pt] (x) at (5.5,2.5) {$x$} ;
		\node[inner sep=1.7pt] (y) at (2.5,5.5) {$y$} ;
		\draw[->, line width=0.4mm] (x) to (y) ; 
		\draw[->, line width=0.4mm] (q) to (x) ; 
		\draw[->, line width=0.4mm] (y) to (q) ; 
	\end{tikzpicture}
	\caption{
		Certificate against $q$ being a strict saddlepoint. Arrows point from larger to smaller entries, inequalities involving $q$ are strict. 
		The conditions of a saddlepoint and $x\ge y$ imply a cycle.
	}
	\label{fig:certificate}
\end{figure}

The process of collecting such certificates for all $n^2$ row-column pairs (through $\OO{n\lg{n}}$ comparisons) now appears very similar to sorting $n$ items, i.e.\ collecting ``ordering certificates'' for all $\binom{n}{2}$ pairs of items. Given these observations, and the fact that the bound of $\OO{n\lg{n}}$ has not been improved in over three decades, it is natural to conjecture that a ``sorting barrier'' holds, rendering the $\OO{n\lg{n}}$ runtime optimal. 

Surprisingly, this is not the case. In this paper we show that the strict saddlepoint of an $n \times n$ matrix (or a certificate of its absence) can be found in \emph{almost linear}, $\OO{n \lg^*{n}}$ worst-case time, where $\lg^*$ denotes the very slowly growing \emph{iterated logarithm} function. Thus, the runtime is in $o(n \lg\lg\dots \lg{n})$ with the $\lg$ function iterated any fixed number of times.

The result is based on the observation that finding a certain \emph{pseudo-saddlepoint} (PSP) of a matrix is sufficient. A PSP always exists, but may not be unique, and may not correspond to a strict saddlepoint (SSP) or even to a general saddlepoint (SP). However, if $A$ has a SP or SSP, then its \emph{value} equals the values of all PSPs of~$A$. Finding a PSP (through a recursive approach) appears to be easier than finding a SSP directly, and having found a PSP, it is easy to locate a SSP with the same value, or to rule out its existence. The recursion can be bootstrapped starting from the Bienstock et al.\ $\OO{n\lg{n}}$ algorithm, obtaining a sequence of algorithms that eventually yield the following general result.

\begin{restatable}{thm}{restatethma}
\label{thm:main}
 Given an $m\times n$ or $n \times m$ matrix $A$, where $m \geq n$, we can determine whether $A$ has a strict saddlepoint, and report such an entry, in $\OO{m \lg^*{n}}$ time.  
\end{restatable} 

Our approach is deterministic, simple (despite the subtle runtime bound)
and operates on the input matrix $A$ only via queries and comparisons between entries. The result is close to optimal: to find the strict saddlepoint, its row and column must be fully queried, yielding an $m+n-1$ lower bound on the number of operations. Whether an algorithm (perhaps randomized) with a linear  runtime exists is an intriguing open question. 

The result is also relevant to the general saddlepoint problem. Locating a (non-strict) saddlepoint, or even deciding if one exists, are subject to the $\Omega(n^2)$ lower bound of Llewellyn et al.~\cite{Llewellyn1988}. However, if it is known that a saddlepoint exists, our algorithm can compute its \emph{value} in the time given in Theorem~\ref{thm:main}. 

After some preliminaries in \S\,\ref{sec2}, we introduce our main algorithm (proving Theorem~\ref{thm:main}) in successive steps, in \S\,\ref{sec3}, \S\,\ref{sec4}, and \S\,\ref{sec_imp}. 
 In \S\,\ref{sec5} we also describe an alternative (perhaps even simpler) approach that already improves upon the result of Bienstock et al., yielding a runtime of $\OO{m\lg\lg{n}}$ for $m \times n$ or $n \times m$ matrices, with $m\geq n$. 

\subparagraph*{Further related work.} Hofri and Jacquet~\cite{HofriJacquet} study the complexity of Knuth's algorithms for matrices with distinct entries that are \emph{randomly permuted}, and Hofri~\cite{Hofri} also studies the distribution of the saddlepoint value and the probability of its existence in matrices with random entries; see also Knuth~\cite[\S\,1.3.2]{Knuth1}. 

\section{Preliminaries} \label{sec2}

Denote $[n] = \{1,\dots,n\}$ and $[a,b] = \{a,a+1,\dots,b\}$. 
%
The binary iterated logarithm $\lg^*{n}$ is defined as 
\begin{align*}
\lg^*{n}=
\begin{cases}
0 & \mathrm{~if~} n \leq 1, \mathrm{~and~}\\ 1+\lg^*{(\lg{n})} & \mathrm{~if~} n>1.
\end{cases}
\end{align*}

Let $A$ be a matrix with $m$ rows and $n$ columns, and let $a_{i,j}$ be the entry of $A$ in row $i$ and column $j$, for all $i \in [m]$ and $j \in [n]$. 

A \emph{saddlepoint} (SP) of $A$ is an entry $a_{i,j}$ such that $a_{i,j} \geq a_{i,k}$ for all $k \in [n] - \{j\}$, and $a_{i,j} \leq {a_{k,j}}$, for all $k \in [m] - \{i\}$. In words, $a_{i,j}$ is the maximum in its row, and the  minimum in its column. If $a_{i,j}$ is a SP with all inequalities being strict, then we call $a_{i,j}$ a \emph{strict saddlepoint} (SSP) of $A$.

Note that the only assumption we make on the entries of the matrix $A$ is that they are from an ordered set, allowing constant-time pairwise comparisons. In particular, we do not require matrix entries to be pairwise distinct.

Not all matrices admit a SSP  or SP (take, for instance, the identity matrix of size $n > 1$), but if it exists, the SSP must be unique. Indeed, suppose that $a_{i,j}$ and $a_{i',j'}$ are two distinct SSPs. The cases $i = i'$ or $j = j'$ result in an immediate contradiction. Otherwise, from the definition of SSP, $a_{i,j} > a_{i,j'} > a_{i',j'} > a_{i',j} > a_{i,j}$, again, a contradiction.\footnote{Note that while at most one entry can be the SSP, its value may appear multiple times in the matrix.}


\subparagraph*{Pseudo-saddlepoint.} We next define a concept that is essential for our algorithms. 

Let $a_{i,r(i)}$ denote a maximum of row $i$ in $A$, i.e.\ $a_{i,r(i)} = \max\{a_{i,1}, \dots, a_{i,n}\}$ for all $i \in [m]$, and let $a_{c(j),j}$ denote a minimum of column $j$ in $A$, i.e.\ $a_{c(j),j} = \min\{a_{1,j}, \dots, a_{m,j}\}$, for all $j \in [n]$ (breaking ties arbitrarily). 

A \emph{pseudo-saddlepoint} (PSP) of $A$ is an entry $a_{i,j} = v$ of $A$ such that every row of $A$ has an entry larger or equal to $v$ and every column of $A$ has an entry smaller or equal to $v$. Equivalently, $a_{i,j} \leq a_{k,r(k)}$ for all $k \in [m]$, and $a_{i,j} \geq a_{c(k),k}$ for all $k \in [n]$. It follows that PSPs are exactly the entries with value in $[C,R]$, where $C$ is the \emph{maximum} of the column-minima $a_{c(k),k}$ and $R$ is the \emph{minimum} of the row-maxima $a_{k,r(k)}$.

\begin{observation}
    Every matrix has at least one PSP.
\end{observation}
\begin{proof}
    It is enough to show $C \leq R$. Indeed, let $C = a_{c(j),j}$ and $R = a_{i,r(i)}$.
    Then, since $C$ ($R$) is the minimum (maximum) of its respective column (row), $a_{c(j),j} \leq a_{i,j} \leq a_{i,r(i)}$.
\end{proof}


The example $\Bigl(\begin{smallmatrix}
0 & 7 & 5\\
6 & 4 & 2\\
3 & 1 & 8
\end{smallmatrix}\Bigr)$
shows that PSPs of a matrix may have different values (here, $C=2$ and $R=6$, so all entries with value in $[2,6]$ are PSPs, but the matrix admits no SSP, and in fact, no SP). The existence of a saddlepoint (strict or not), however, determines the value of all PSPs, as we show next.

\begin{lemma}\label{lemma:strict_is_pseudo}
    If a matrix $A$ has a SP of value $s$, then every PSP of $A$ has value $s$.
\end{lemma}
\begin{proof}
    Let $a_{i,j} = s$ be a SP of $A$. Then, 
    $s \leq a_{c(j),j} \leq C$, by the definition of a SP and by the fact that $C$ is the maximum of the column-minima. Thus, 
    no value $v<s$ can correspond to a PSP. Symmetrically, $R \leq a_{i,r(i)} \leq s$, so no value $v > s$ can correspond to a PSP.
\end{proof}

\subparagraph*{Testing for a strict saddlepoint.} The next ingredient of our approach is a feasibility test.

\begin{lemma}\label{lem:feas}
Given an $m \times n$ matrix $A$ and a value $s$, we can find, in time $O(m+n)$, a SSP of $A$ of value $s$, or report that no such SSP exists. 
\end{lemma}

\begin{proof}
We perform two searches, both starting at location $(1,1)$ of the matrix and proceeding in a monotone staircase fashion as follows. We call the first the \emph{horizontal search}, and the second the \emph{vertical search}, illustrated in Figure~\ref{fig:staircase}.

\begin{figure}[tbh]
\centering
\includegraphics[scale=0.34]{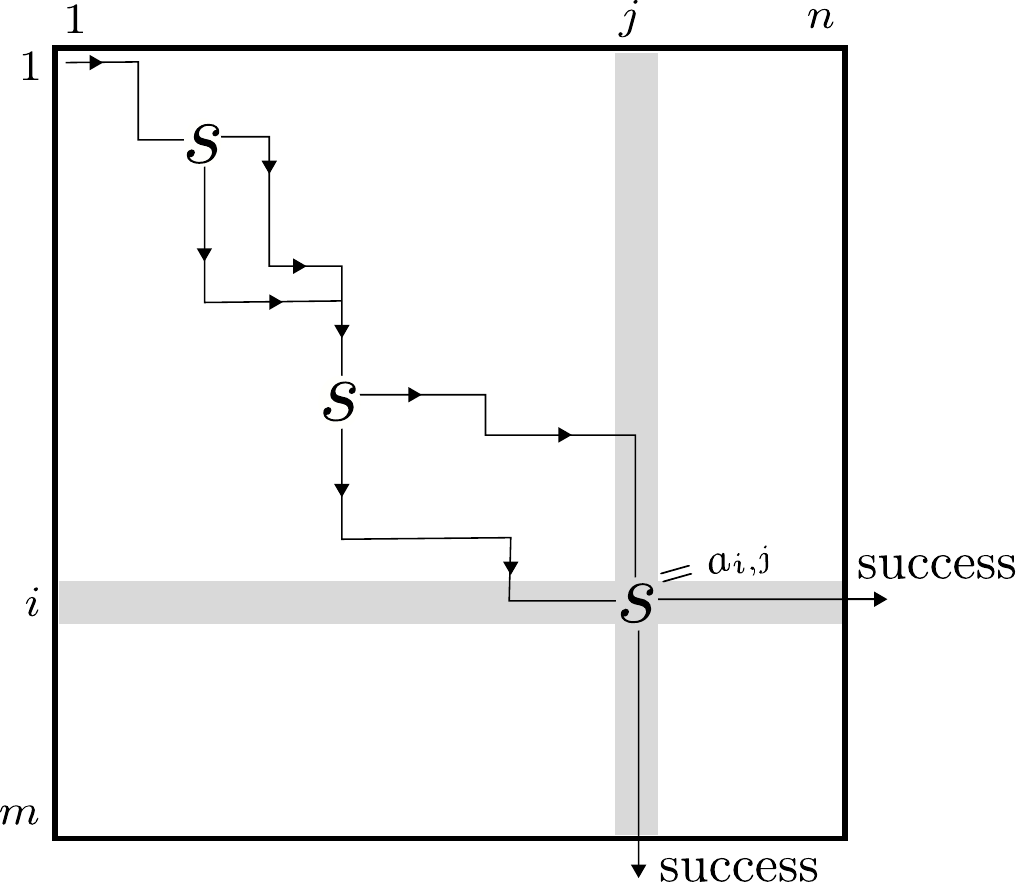}
\caption{\label{fig:staircase}Horizontal and vertical searches in the proof of Lemma~\ref{lem:feas}, for value $s$, with SSP $a_{i,j} = s$. Observe that the two paths can meet at arbitrary entries, but diverge only at entries with value $s$. After finding the SSP, the two searches only proceed horizontally, resp.\ vertically.}
\end{figure}

\emph{Horizontal search}: Suppose we are at $(i,j)$. If $i \in [m]$ and $j\in [n]$ then let $q = a_{i,j}$:

\noindent If $s < q$, set $i\gets i+1$, and continue the search. Otherwise, set $j \gets j+1$, and continue the search. If $i>m$, halt with \emph{failure}, if $j>n$, halt with \emph{success}.

\emph{Vertical search}: Suppose we are at $(i,j)$. If $i \in [m]$ and $j\in [n]$ then let $q = a_{i,j}$:

\noindent If $s \leq q$, set $i\gets i+1$, and continue the search. Otherwise, set $j \gets j+1$, and continue the search. If $i>m$, halt with \emph{success}, if $j>n$, halt with \emph{failure}.

Observe that the two searches differ only in the tie-breaking in the case $s=q$ and whether they report success on exiting on the horizontal or vertical end of the matrix.

If any of the two searches halt with failure, return \emph{false}. Otherwise, suppose the horizontal search exited at $(i,n+1)$ and the vertical search exited at $(m+1,j)$. Then, compare $a_{i,j}$ with all other entries in row $i$ and in column $j$, returning \emph{true} if $a_{i,j}$ is a SSP, and otherwise returning \emph{false}. 

Clearly, the runtimes of both searches and the final verification are $\OO{m+n}$. 
It remains to show, towards establishing Lemma~\ref{lem:feas}, that the algorithm correctly identifies the SSP.

Due to the final verification step, the algorithm cannot falsely report a SSP. 
So suppose that there is a SSP $a_{i,j}$ of $A$. Then, both searches eventually reach an entry $(i,j')$ with $j' \leq j$ or an entry $(i',j)$ with $i' \leq i$. From then on, the searches proceed directly to $a_{i,j}$. The horizontal search then executes only $j \gets j+1$ steps, eventually returning success, and the vertical search executes only $i \gets i+1$ steps, eventually returning success. As the verification succeeds, the algorithm correctly identifies $a_{i,j}$ as the SSP.
\end{proof}

We remark that in the algorithm of Lemma~\ref{lem:feas}, when searching for a value $s$, the horizontal and vertical searches cannot both fail. Indeed, if the horizontal search exits ``at the bottom'' (at $(m+1,j)$) and the vertical search exits ``to the right'' (at $(i,n+1)$), then the two paths must diverge at an entry of value $s$ ``in the wrong way'', i.e.\ the horizontal search moving vertically, and the vertical search moving horizontally, which is impossible. By distinguishing the other cases, we can turn the test into a parametric search tool, that will be useful in \S\,\ref{sec5}.

\begin{observation}\label{obs:search}
If in the algorithm of Lemma~\ref{lem:feas}, when searching for a value $s$, the horizontal search fails, then the SSP (if exists) must be greater than $s$. If the vertical search fails, then the SSP (if exists) must be smaller than $s$. If both searches succeed and the final test fails, then the matrix has no SSP.
\end{observation}
\begin{proof}
If the horizontal search fails, then it has identified in each row an entry of value greater than $s$, thus the SSP must be greater than $s$. The second case is symmetric. In the third case, we learn that each row has an entry of value at least $s$ and each column has an entry of value at most $s$. Thus, the SSP could only have value $s$.\footnote{We also mention the somewhat counter-intuitive fact that for every matrix there is exactly one value for which both searches succeed (regardless of whether the matrix has a SSP or SP), and the procedure can be seen as searching for this value. As we lack an immediate use for this fact, we omit the (easy) proof.} 
\end{proof}

\subparagraph*{Reduction to square matrices.}
We briefly argue that when computing PSPs, it is sufficient to focus on square matrices. Suppose, more generally, that the input matrix $A$ is of size $m \times n$. Assume, w.l.o.g.\ that $m \geq n$, otherwise let $A = -A^{T}$, without affecting the structure of SSP or PSPs. (Explicitly transposing and negating $A$ would of course be too costly, but we can do this ``on demand'', only for the queried entries.) 

We reduce the computation of a PSP of $A$ to computations on square matrices (similar arguments were used by Llewellyn et al.\ \cite{Llewellyn1988} and Bienstock et al.\ \cite{Bienstock1991} for the SSP).
\begin{lemma}\label{lemma:general_to_square}
    Computing a PSP of an $m \times n$ matrix with $m > n$ can be reduced to the computation of PSPs of $\left\lceil m/n \right\rceil$ matrices of size $n\times n$, with an additive overhead of $\OO{m/n}$ on the runtime.
\end{lemma}
\begin{proof}
Let $A$ be an $m\times n$ matrix with $m > n$. Divide $A$ into $\left\lceil m/n \right\rceil$ (possibly overlapping) matrices of size $n\times n$, compute a PSP for each of them, and return the one with minimum value $v$. Because that entry is a PSP of its corresponding matrix $A'$, every column of $A'$ (and thus every column of $A$) has an entry smaller or equal to $v$. Every row of $A$ has a value larger or equal to the computed PSP of one of the smaller matrices containing this row, which is itself larger or equal to $v$. Thus, every row of $A$ contains a value larger or equal to $v$.
\end{proof}

\begin{corollary}\label{cor}
Given an algorithm that finds a PSP in an $n \times n$ matrix in time $\OO{n\cdot f(n)}$ for arbitrary $f(n) \geq 1$, we can compute a PSP in an $m \times n$ or $n \times m$ matrix with $m \geq n$ in time $\OO{m \cdot f(n)}$.
\end{corollary}

\section{A baseline algorithm for pseudo-saddlepoints} \label{sec3}

In this section we adapt the algorithm of Bienstock et al.\ \cite{Bienstock1991} for finding the SSP, to find a PSP. For completeness, we repeat the analysis, although the modifications needed are minor.

Let $H = \{(r_1,c_1,v_1), (r_2,c_2,v_2), \ldots, (r_q,c_q,v_q)\}$ be a set of $q$ triplets of the form (\textit{row, column, value}), corresponding to entries in the input matrix~$A$, with $v_1\leq v_2\leq \cdots \leq v_q$, and satisfying the following three properties:
\begin{itemize}
    \item P1. $H$ has at most one entry from each row or column of $A$.
    \item P2. Every row which does not appear in $H$ has an entry larger or equal to $v_q$.
    \item P3. Every column which does not appear in $H$ has an entry smaller or equal to $v_1$.
\end{itemize}

\begin{lemma}\label{lem:op}
    Let $(i,j,a_{i,j}) = (r_1,c_1,v_1)$ and $(k,\ell,a_{k,\ell}) = (r_q,c_q,v_q)$ be the elements of $H$ with minimum and maximum value respectively. By querying $a_{i,\ell}$ and doing a constant number of comparisons and insertions/deletions in $H$, we can reduce the size of $H$ by one while preserving properties \emph{P1}, \emph{P2}, and \emph{P3}.
\end{lemma}
\begin{proof}
    We know from P1 that $i\neq k$ and $j\neq \ell$. We query $a_{i,\ell}$ and distinguish three cases based on the comparisons with $a_{i,j}$ and $a_{k,\ell}$. Recall that $a_{i,j} = v_1 \le v_q=a_{k,\ell}$.
    
    Case 1: $a_{i,\ell} \leq a_{i,j} \leq a_{k,\ell}$.
    Then row $k$ has an entry larger or equal to $v_{q-1}$ (namely $v_q = a_{k,\ell}$). By P2 and the fact that $v_{q-1}\leq v_q$, every row which does not appear in $H$ also has an entry larger or equal to $v_{q-1}$.  Column $\ell$ has an entry smaller or equal to $v_1$ (namely $a_{i,\ell}$) and every column which does not appear in $H$ also has an entry smaller or equal to $v_1$ (by P3). Thus, we can delete $(r_q,c_q,v_q)$ from $H$ while preserving all three properties.
    
    Case 2: $a_{i,j} \leq a_{k,\ell} \leq a_{i,\ell}$. By a symmetric argument we can delete $(r_1,c_1,v_1)$ from $H$ while preserving all three properties.

    Case 3: $a_{i,j} < a_{i,\ell} < a_{k,\ell}$. Let $u = \min\{v_2,a_{i,\ell}\}$ and $v = \max\{v_{q-1},a_{i,\ell}\}$.
    The row $k$ has an entry larger or equal to $v$ (namely $a_{k,\ell}$). By P2 this is also the case for every row which does not appear in $H$. The column $j$ has an entry smaller or equal to $u$ (namely $a_{i,j}$). By P3 this is also the case for every column which does not appear in $H$. Thus, we can insert $(i,\ell,a_{i,\ell})$ and delete both $(r_1,c_1,v_1)$ and $(r_q,c_q,v_q)$ from $H$ while preserving all three properties.
\end{proof}

Given an $n\times n$ matrix $A$, we start by setting $H=\{(1,1,a_{1,1}), (2,2,a_{2,2}), \ldots, (n,n,a_{n,n})\}$. We then repeatedly apply Lemma~\ref{lem:op} while maintaining $H$ as a heap or a dynamically balanced binary search tree ordered by entry values (to guarantee $\OO{\log n}$ time per application of the lemma), until $H$ has only one entry. By properties P2 and P3, this entry will be a PSP. We obtain the following theorem.



\begin{theorem}\label{thm:base_algo}
    Given an $n\times n$ matrix $A$, we can report a pseudo-saddlepoint (PSP) of $A$ in time $\OO{n\log n}$, by querying at most $2n-1$ entries of $A$. 
\end{theorem}

\section{Bootstrapping the algorithm} \label{sec4}

We now improve the algorithm of Theorem~\ref{thm:base_algo}, by embedding it into a recursive approach. The following lemma allows the decomposition of the input matrix, which is the key step in our subsequent algorithm. 

\begin{lemma}\label{lemma:block_matrix}
    Let $A$ be an $n\times n$ matrix. Let $R_1, R_2, \ldots, R_k$ and $C_1, C_2, \ldots, C_\ell$ be sets of indices, so that $R_1 \cup \cdots \cup R_k = [n] = C_1 \cup \cdots \cup C_\ell$. For all $i \in [k]$ and all $j \in [\ell]$, let $A_{R_i,C_j}$ denote the submatrix of $A$ obtained by taking the rows and columns of $A$ with indices in $R_i$, resp.\ $C_j$. Let $A'$ be a $k\times \ell$ matrix whose entry $a'_{i,j}$ is a PSP of $A_{R_i,C_j}$ for all $i \in [k]$ and $j \in [\ell]$.
    Then all PSPs of $A'$ are PSPs of $A$.
\end{lemma}
\begin{proof}
    Let $v$ be the value of a PSP of $A'$. By definition, for every $i \in [k]$, there is some $j \in [\ell]$, so that $a'_{i,j} \geq v$. Because $a'_{i,j}$ is a PSP of $A_{R_i,C_j}$, every row  of $A$ with index in $R_i$ has an entry larger or equal to $a'_{i,j}$. Thus, every row of $A$ has an entry larger or equal to $v$. A symmetric argument shows that every column of $A$ has an entry smaller or equal to $v$.
\end{proof}

We are ready to describe the main subroutine of our algorithm. For ease of presentation we start with a simpler version that already has an almost linear runtime bound. 

\begin{theorem}\label{thm:pseudo_logstar}
    Given an $n\times n$ matrix $A$, we can report a pseudo-saddlepoint (PSP) of $A$ in $\OO{n \cdot 2^{\lg^*{n}}}$ time.
\end{theorem} 
\begin{proof} 
    The algorithm works as follows.
    \begin{itemize}
        \item If $n=1$, return the only entry of $A$.
        \item Otherwise, let $\ell = \ceil{\lg{n}}$, and divide the rows and columns of $A$ each into $\bigl\lceil\frac{n}{\ell}\bigr\rceil$ (possibly overlapping) intervals of size $\ell$. This divides $A$ into $\bigl\lceil{\frac{n}{\ell}\bigr\rceil}^2$ (possibly overlapping) square matrices of size $\ell \times \ell$. We obtain a new matrix $A'$ by conceptually replacing each smaller matrix with a PSP of that matrix.
        \item Run the algorithm of Theorem \ref{thm:base_algo} on $A'$. Each time a new entry of $A'$ is queried, run the current algorithm recursively on the corresponding submatrix of $A$ to obtain the sought value.
    \end{itemize}

    The correctness of the algorithm follows directly from Lemma~\ref{lemma:block_matrix}. 
    Let us turn to the runtime analysis. In the first level of the recursion, the current algorithm runs the algorithm of Theorem \ref{thm:base_algo} on a $\bigr\lceil{\frac{n}{\ceil{\lg n}}}\bigr\rceil\times \bigr\lceil{\frac{n}{\ceil{\lg n}}}\bigr\rceil$ matrix (costing $O(n)$ time) and makes $2\bigl\lceil{\frac{n}{\ceil{\lg n}}}\bigr\rceil-1$ recursive calls on matrices of size $\ceil{\lg n}\times \ceil{\lg n}$. 
    Thus, for all $n\geq 2$ and a large enough constant $c$, the runtime $T(n)$ of the algorithm obeys
    \[T(n) \leq cn + \left(2\ceil{\frac{n}{\ceil{\lg n}}}-1\right)T(\ceil{\lg n}).\]

We show by induction that  $T(n) \leq 2cn \cdot 2^{\lg^* n}-3cn\lg^*n$ for all $n\geq 1$, and thus $T(n) \in \OO{n \cdot 2^{\lg^*{n}}}$. For $1\leq n \leq 16$, the bound holds assuming $c$ is large enough. Now assume that $n > 16$ and that the result is true for all values smaller than $n$. We have:
    \begin{align}
        T(n) &\leq cn + \left(2\ceil{\frac{n}{\ceil{\lg n}}}-1\right)T(\ceil{\lg n}) \notag\\
        &\leq cn + \left(2\ceil{{\frac{n}{\ceil{\lg n}}}}-1\right) \left(2c\ceil{\lg n}\cdot 2^{\lg^*(\ceil{\lg n})}-3c\ceil{\lg n}\lg^*(\ceil{\lg n}) \right)\\
        &\leq cn + \left(\frac{2n}{\ceil{\lg n}}+1\right) \left(2c\ceil{\lg n}\cdot 2^{\lg^*(n)-1}-3c\ceil{\lg n}(\lg^*(n)-1) \right)\\
        &\leq 2cn \cdot 2^{\lg^*{n}} - 3cn\lg^*n - 3cn\lg^*n +7cn+c\ceil{\lg n} \cdot 2^{\lg^*{n}} \\
        &\leq 2cn \cdot 2^{\lg^*{n}} - 3cn\lg^*n - 12cn +7cn+5cn \\
        &\leq 2cn \cdot 2^{\lg^* n}-3cn\lg^*n. \notag
    \end{align}
Here,
\begin{itemize}
    \item (1) follows by induction,
    \item (2) uses that $\ceil{x} \leq x+1$ for all $x$, and that $\lg^*(\ceil{\lg n}) = \lg^*{n}-1$ for $n \geq 2$,
    \item (3) follows 
    via simple manipulation and dropping a negative term, 
    \item (4) uses the facts that $\ceil{\lg n}\cdot 2^{\lg^* n} < 5n$ for $n\geq 1$, and $\lg^* n \geq 4$ for $n>16$. \qedhere
\end{itemize}
\end{proof}

We remark in passing that an early stopping of the recursion would yield for all $k \in O(\lg^*{n})$, a runtime of $n \cdot 2^{\OO{k}}\lg^{(k)} n$ with only $n \cdot 2^{\OO{k}}$ entries of $A$ queried.

\medskip

Our overall algorithm is as follows: Given an input matrix $A$ of size $m \times n$ with $m \geq n$, first find a PSP $s$ of $A$ (via Theorem~\ref{thm:pseudo_logstar} and Corollary~\ref{cor}) in time $\OO{m \cdot 2^{\lg^*{n}}}$. Then, verify whether $A$ admits a SSP of value $s$ (via Lemma~\ref{lem:feas}) in time $\OO{m+n}$. If yes, then report it, if not, then conclude (by Lemma~\ref{lemma:strict_is_pseudo}) that $A$ has no SSP. This yields  an overall runtime of $\OO{m \cdot 2^{\lg^*{n}}}$.

\subparagraph*{Remark.} Our algorithm can also be used for computing the \emph{value} of the (non-strict) saddlepoint (SP), assuming that it exists, within the same runtime: we simply find a PSP $s$ of $A$ (via Theorem~\ref{thm:pseudo_logstar} and Corollary~\ref{cor}), and conclude (by Lemma~\ref{lemma:strict_is_pseudo}) that the SP value is~$s$. \emph{Locating} a SP entry requires quadratic time in the worst case~\cite{Llewellyn1988}. We can improve this, however, if the SP value $s$ appears only few times in $A$. More precisely, we can locate a SP of value $s$ in an $m \times n$ matrix $A$ with $k$ entries of value $s$ in $\OO{k(m+n)}$ additional time. 

The approach is as follows: Run the horizontal search of Lemma~\ref{lem:feas} with value $s$. The search necessarily succeeds, finding in each column an entry of value at most $s$. The SP must be in a column where we encountered the value $s$. There are at most $k$ such columns, so look through all of them in $O(mk)$ time to collect all candidate entries of value $s$ (again, at most $k$ of them). Test the candidates in $\OO{m+n}$ time each, for a total of $\OO{k(m+n)}$.
 
\section{Improved runtime}\label{sec_imp} 
A closer look at the baseline algorithm of Theorem~\ref{thm:base_algo} reveals that only $n-1$ of the $2n-1$ queried entries are chosen adaptively during runtime; the remaining $n$ entries are on the main diagonal $(a_{i,i})_{i \in [n]}$ of the input matrix $A$. A similar observation applies to the algorithm of Theorem~\ref{thm:pseudo_logstar} described in \S\,\ref{sec4}. Denoting $\ell = \ceil{\lg{n}}$, here $\bigr\lceil{\frac{n}{\ell}\bigr\rceil}$ of the 2$\bigr\lceil{\frac{n}{\ell}\bigr\rceil}-1$ recursive calls are for submatrices whose position is fixed upfront. 

This suggests an improvement to the algorithm of Theorem~\ref{thm:pseudo_logstar}, by solving $\bigr\lceil{\frac{n}{\ell}\bigr\rceil}$ of the subproblems directly, in an $O(n)$ time preprocessing step, and thereby reducing the number of recursive calls. 

Let us first fix the decomposition of the input matrix $A$ into submatrices of size $\ell \times \ell$ as follows. For $i = 0, \dots, \floor{{n}/{\ell}}-1$, let $R_{i} = C_{i} = [i\cdot \ell+1, (i+1)\cdot \ell]$, and set the last (possibly overlapping) interval $R_{\floor{{n}/{\ell}}} = C_{\floor{{n}/{\ell}}} = [n-{\ell}+1, n]$. We have $\cup_i{R_i} = \cup_i{C_i} = [n]$.

\begin{lemma}\label{lem_imp}
Given an $n \times n$ matrix $A$, we can transform it in time $O(n)$ into an $n \times n$ matrix $B$, so that a PSP of value $q$ of $B$ implies a PSP of value $q$ of $A$. Moreover, in $O(n)$ time we can compute a PSP of every ``diagonal box'' $B_{R_i,C_i}$. 
\end{lemma}

Lemma~\ref{lem_imp} serves as a preprocessing step for each call of the algorithm of Theorem~\ref{thm:pseudo_logstar}.
With this preprocessing, the algorithm of Theorem~\ref{thm:pseudo_logstar} needs to make only $\ceil{\frac{n}{\ell}}-1$ recursive calls on matrices of size $\ell \times \ell$. As far as these recursive calls are concerned, the preprocessed matrix $B$ is identical to $A$, up to permuting rows and columns, which can be maintained using straightforward bookkeeping. 

Thus, for all $n \geq 2$ and a large enough constant $c'$, the recurrence for the runtime  becomes 
\[T(n) \leq c'n + \left(\ceil{\frac{n}{\ell}}-1\right)T(\ell).\]

We show by induction that  $T(n) \leq 2c'n \lg^* n$ for all $n\geq 1$, and thus $T(n) \in \OO{n \lg^*{n}}$: 
    \begin{align}
        T(n) &\leq c'n + \left(\ceil{\frac{n}{\ell}}-1\right)T(\ell) \notag\\
        &\leq c'n + \left(\ceil{{\frac{n}{\ell}}}-1\right) \left(2c'\ell \lg^*{\ell} \right) \notag\\
        &\leq c'n + 2c'n  ( \lg^*(n)-1) \notag\\
        &\leq 2c'n \lg^*{n}. \notag
    \end{align}

Together with Corollary~\ref{cor} and Lemma~\ref{lem:feas} and Lemma~\ref{lemma:strict_is_pseudo}, this implies our main result.

\restatethma*

Here again one could stop the recursion early, yielding for all $k \in O(\lg^*{n})$, a runtime of $O(n\lg^{(k)}n + nk)$ with only $\OO{nk}$ entries of $A$ queried.

\medskip

It remains to describe and analyze the preprocessing step.

\begin{proof}[Proof of Lemma~\ref{lem_imp}]

Given a square input matrix $A$, consider the following transformation, written as an in-place procedure, that results in a matrix $B$ of the same size.

\medskip

\begin{algorithmic}[1]
\Statex \textbf{Transform$(A,t)$} 
\Statex \textbf{Input:} an $n \times n$ matrix $A$ and stopping threshold $t$.
\State \textbf{if {$n \leq t$} } \textbf{then halt} 
\State $v \gets $ \textbf{Select}$(\{a_{n,1}, a_{n-1,2}, \dots, a_{1,n}\}, \ceil{n/2})$ 
\State \textbf{Partition} $(a_{n,1}, a_{n-1,2}, \dots, a_{1,n})$ around $v$
\State $a_{i,i} \gets v$ for $i=1, \dots, \ceil{n/2}$
\State \textbf{Transform$(A_{[\ceil{n/2}+1,n],[\ceil{n/2}+1,n]},t)$}
\end{algorithmic}

\medskip

The transformation works as follows: select the median (element of rank $\ceil{n/2}$) $v$ of the antidiagonal $\{a_{n,1}, a_{n-1,2}, \dots, a_{1,n}\}$.
Then, partition the antidiagonal around $v$ as detailed below, so that $v$ goes into position $a_{n-\ceil{n/2}+1, \ceil{n/2}}$ of the matrix, with entries smaller on its left and entries larger on its right. 

Then, set the first $\ceil{n/2}$ entries on the main diagonal to $v$. Finally, repeat the transformation recursively on the bottom right quadrant of the matrix, starting from an entry of the main diagonal. Stop when the matrix size falls below a stopping threshold $t$. The initial call is Transform$(A,2\lg{n})$, to preprocess a matrix $A$ of size $n \times n$, with stopping threshold $t = 2\lg{n}$; the early stopping is to avoid affecting the rightmost (overlapping) boxes of $A$. The effect of the transformation is illustrated in Figure~\ref{fig:transform}(a).

\begin{figure}[tbh]
\centering
\includegraphics[height=2in]{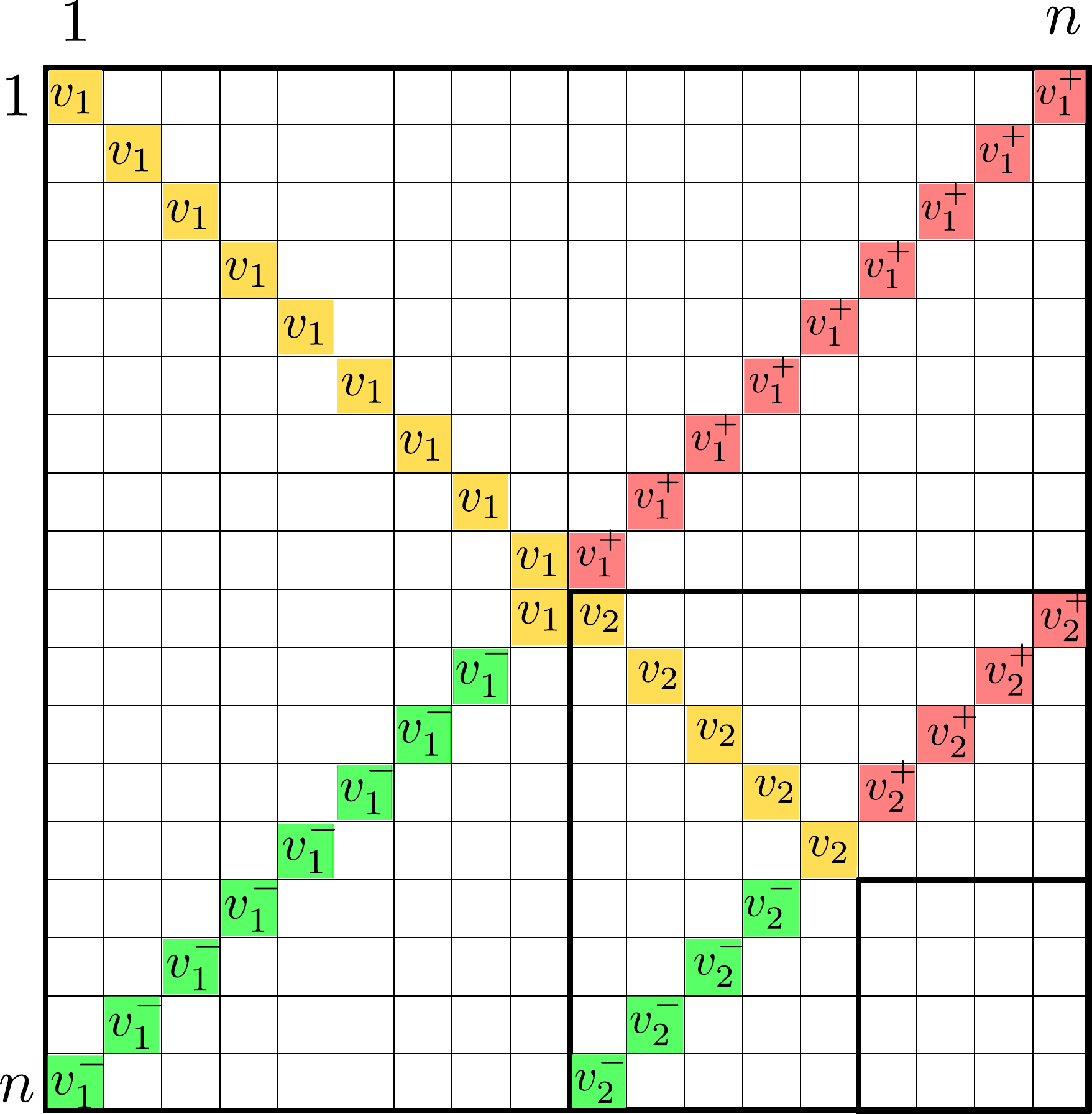}~~~~
\includegraphics[height=2in]{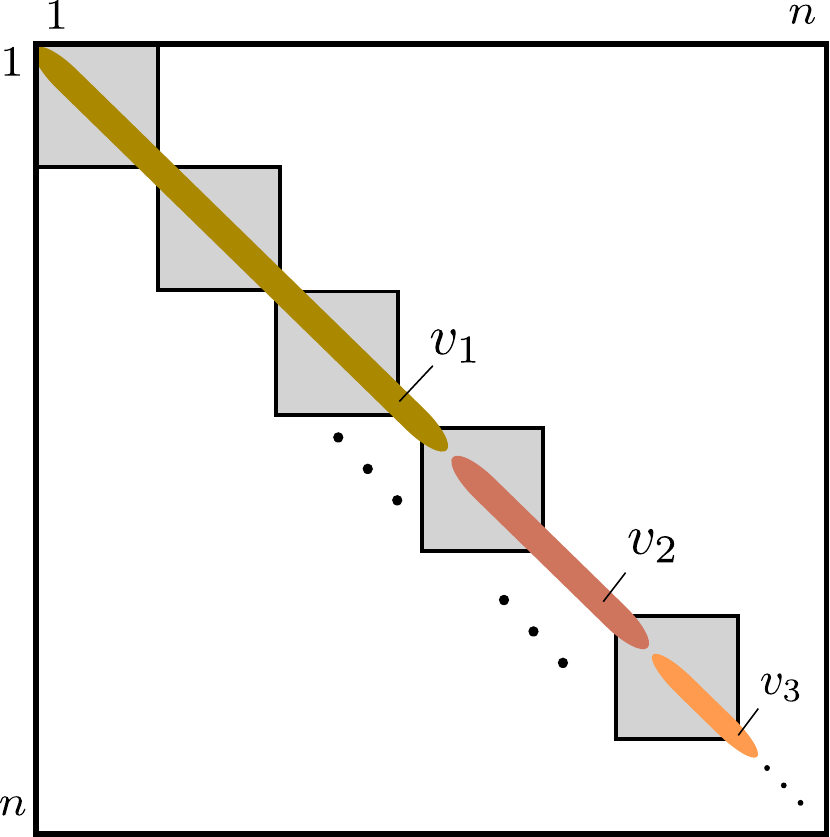}
\caption{\label{fig:transform} (\emph{a}) The effect of procedure Transform on an input matrix $A$ of size $18 \times 18$: the antidiagonal is partitioned around the median, and the median value is copied to half of the main diagonal. The procedure is repeated recursively on the lower right quadrant. Values $v_1, v_2, \dots$ denote the medians found in subsequent calls, and $v_i^-$, $v_i^+$ indicate values $\leq v_i$, resp.\ $\geq v_i$.\\ Note that swaps during the partitioning step of recursive call $i$ may move elements away from the antidiagonal from call $j < i$. For simplicity, this is not reflected in the figure. \\
(\emph{b}) Solving the diagonal-subproblems in the algorithm of Theorem~\ref{thm:pseudo_logstar}. Squares indicate $\ceil{\lg{n}} \times \ceil{\lg{n}}$ size subproblems. All, but $\lg{n}+\OO{1}$ of these have uniform diagonal. }
\end{figure}

Let us first argue that the transformation takes linear time.  Indeed, line 2 can employ linear time selection (e.g.\ \cite{median}), and line 3 can be achieved by simulating a standard partitioning procedure on the antidiagonal, swapping pairs of entries by swapping the corresponding pairs of rows and columns. Note that when running line 3 in a recursive call of Transform, we still swap pairs of rows or columns of the full matrix.  Notice that these operations can be implemented with simple bookkeeping in constant time per swap. Line 4 requires to copy the median element along the diagonal in $\ceil{n/2}$ locations (this can be achieved e.g.\ by saving the diagonal into an array). Thus, lines 1--4 take $cn$ time, for a sufficiently large $c$. With the recursive call in line 5, the total runtime of {Transform}$(A)$ can be bounded as $c(n + n/2 + n/4 + n/8 + \cdots) \leq 2cn \in \OO{n}$. 

\medskip

We next argue that a PSP of the transformed matrix $B$ implies a PSP of $A$ of the same value. Recall that $q$ is a PSP-value if and only if $q \in [C, R]$, where $C$ is the maximum of the
column-minima and $R$ is the minimum of the row-maxima.

Initially $B=A$, so the PSPs are the same. Consider a call of Transform at an arbitrary level of recursion. Swapping pairs of rows or columns does not affect the PSP values, thus the claim holds for lines 1--3. Copying the median value $v$ in line 4 cannot decrease $C$, since the affected column already contains a value at most $v$ due to the partitioning step. Similarly, it cannot increase $R$, since the affected row already contains a value at least $v$ (see Figure~\ref{fig:transform}(a)). Thus, denoting by $C'$, $R'$ the new maximum of column-minima, resp.\ minimum of row-maxima, we have $[C',R'] \subseteq [C,R]$, so no new PSPs are created.

\medskip

It remains to compute a PSP for each diagonal box $B_{R_i,C_i}$. Notice that the diagonals of these boxes coincide with the diagonal of the full (preprocessed) matrix $B$. Moreover, for all but at most $\lg{n} + \OO{1}$ of the boxes $B_{R_i,C_i}$, their diagonal contains a single value. This is because, after preprocessing, the diagonal of $B$ consists of at most $\lg{n}$ contiguous uniform sections (corresponding to calls of Transform), and a last section of length at most $2\lg{n}$, that was unaffected by Transform. Only boxes that intersect with the boundaries between sections, and a constant number of boxes at the end will have a non-uniform diagonal; see Figure~\ref{fig:transform}(b) for an illustration. 

Notice that if every diagonal entry of a matrix is $v$, then $v$ is a PSP of the matrix. Thus, all but $\lg{n}+\OO{1}$ of the diagonal boxes have their PSP readily available. For the remaining diagonal boxes, we call the baseline algorithm of Theorem~\ref{thm:base_algo}, adding a term $\OO{\lg{n} \cdot \lg{n}\lg\lg{n}} \subset \OO{n}$ to the runtime. The total runtime is $O(n)$, finishing the proof. 
\end{proof}

\section{An alternative approach} \label{sec5}

In this section we briefly describe an algorithm for the SSP problem with a runtime of $O(n \lg \lg n)$ on an $n \times n$ matrix $A$.  
While the bound is weaker than the previous ones, we find the approach worth mentioning due to its simplicity, and since it can run faster on certain inputs. 

The algorithm has two phases. The \emph{first phase} makes use of Observation~\ref{obs:search} 
to find progressively better upper and lower bounds on
the value of the SSP (should it exist), and thereby eliminate rows or columns of $A$ that cannot contain the SSP.
The \emph{second phase} begins when one side-length of the matrix has been reduced to
at most $\frac{n}{\lg n}$. It makes use of a heap, similarly to the algorithm of Theorem~\ref{thm:base_algo}, but with a different purpose: to also reduce the \emph{longer side} of the matrix to about $\frac{n}{\lg n}$. When both sides of the matrix have length $O{(\frac{n}{\lg{n}})}$, we can finish the job with the baseline algorithm of Theorem~\ref{thm:base_algo}, with a total runtime of $O(n)$. If a SSP is found, we perform an additional $O(n)$ time test with it on the original matrix, to rule out a false positive.

Recall that the search procedure of Lemma~\ref{lem:feas} returns (by Observation~\ref{obs:search}), in $\OO{m+n}$ time, for an $m \times n$ or $n \times m$ matrix $A$ and a search value $s$, one of four answers: (1) $A$ has the SSP $a_{i,j} = s$; (2) $A$ has no SSP; (3) the SSP of $A$ (if exists) is $>s$; or (4) the SSP of $A$ (if exists) is $<s$.

\subparagraph*{First phase.} Let $A$ denote the current $m' \times n'$ matrix, where $n \geq m' \geq n'$ (the other case is symmetric). Assume that entries of $A$ are $a_{i,j}$ with $i \in [m']$, resp.\ $j \in [n']$. 

Compute the median $v$
of the set of elements $D = \{ a_{i, \lceil \frac{i \cdot n'}{m'} \rceil} \mid i \in [m'] \}$ and 
search $A$ with $v$ using the procedure of Lemma~\ref{lem:feas}. If we learn that the SSP can only be larger than $v$, then
recurse on $A_{[m'], [n']-C'}$ where $C' = \{ j \mid a_{i,
j} \in D \land v \geq a_{i, j} \}$. If we learn that the SSP can only be smaller than $v$, then 
recurse on $A_{[m'] - R', [n']}$ where $R' = \{ i \mid a_{i, j} \in D \land
v \leq a_{i, j} \}$. If we find the SSP, or learn that a SSP does not exist, then halt accordingly. 

Observe that the reductions are justified. Indeed, if the SSP must be larger than $v$, then it cannot be in columns that contain entries smaller or equal to $v$, and hence, these columns can be deleted. If the SSP must be smaller than $v$, then it cannot be in rows that contain entries larger or equal to $v$, and hence, these rows can be deleted.

If either side-length of the matrix is reduced to $\frac{n}{\lg n}$, then proceed to the next phase.
\subparagraph*{Second phase.}
Let $A$ denote the current $m' \times n'$ matrix, where $n \geq m' \geq n'$ (the other case is symmetric). Assume $n' \leq \frac{n}{\lg{n}}$, and $m' > 4n'$ (otherwise we can stop). For each column $j \in [n']$, select $\lfloor \frac{m'}{2n'} \rfloor$ rows $R_j \subseteq [\floor{m'/2}]$, so that $R_j \cap R_{j'} = \emptyset$ whenever $j \neq j'$. Observe that initially only the first half of the $m'$ rows are picked. For each column $j$ calculate the minimum $m_j = \min\{a_{i,j} \mid i \in R_j\}$.

Insert the minima $m_1, \dots, m_{n'}$ into a max-heap. For $n'$ iterations, extract the maximum element from the heap. Suppose the currently extracted maximum is $m_j = a_{i,j}$. Then for all $i' \in R_j-\{i\}$, delete row $i'$ of the matrix $A$. 

This is justified, since a SSP cannot exist in row $i'$ of the matrix. Indeed, if such a SSP $a_{i', k}$ existed, then either (1) $k=j$ and we have a contradiction since $m_j \leq a_{i', j}$ (by the choice of $m_j$ as the minimum), or (2) $k\neq j$ which is a contradiction since $m_k \leq m_j \leq a_{i', j} < a_{i', k}$ (by the maximality of $m_j$ in the heap and the row condition of SSPs). Either case contradicts that $a_{i', k}$ is a SSP since it is not strictly smaller than some value in its column.

Now, in the column $j$, select $\lfloor \frac{m'}{2n'} \rfloor$ new elements with row index set $R_j \subseteq [m']$ disjoint from all other $R_{j'}$ with $j \neq j'$. 
Compute the new minimum $m_j = \min\{a_{i,j} \mid i \in R_j\}$, and reinsert $m_j$ into the max-heap.

In $n'$ iterations,
the process removes a constant fraction of rows. Repeat the phase $\OO{\lg\lg{n}}$ times, to reduce the number of rows to $m' \leq 4n'$. 

\subparagraph*{Running time.} Starting with an initial $n \times n$ matrix, every iteration of the first phase runs in $\OO{n}$ time and removes at least a 
constant fraction of the remaining rows or columns, which implies that in $\OO{n \lg \lg n}$ time at least one
side will be reduced to length $\frac{n}{\lg n}$.\footnote{Note that if iterations alternate between removing rows and columns, then the 
runtime can be described by a geometric series that evaluates to $\OO{n}$.}

In the second phase, initializing $R_j$ and $m_j$ for each column $j$ takes $\OO{n}$ total time. Then, each of the $n'$ iterations involve a constant number of heap operations of cost $\OO{\lg{n}}$, and a constant cost per the removal of each row. Since $n' \in O{(\frac{n}{\lg{n}})}$, the execution of the phase takes $\OO{n}$ total time. Repeated $\OO{\lg \lg n}$ times, this yields the total runtime of $\OO{n \lg \lg n}$.

\bigskip

\bibliography{saddlepoints}
\end{document}